\documentclass[a4paper,11pt]{amsart}
\usepackage{amssymb,amsfonts,amsxtra,amscd,mathrsfs,stmaryrd}
\usepackage[all]{xy}
\usepackage{fullpage}
\usepackage{graphicx}
\theoremstyle{plain}
\newtheorem{theorem}{Theorem}[section]
\newtheorem{lemma}[theorem]{Lemma}

\theoremstyle{definition}
\newtheorem{defi}[theorem]{Definition}

\theoremstyle{remark}
\newtheorem{rem}[theorem]{Remark}
\numberwithin{equation}{section}
\newcommand{\smth}[1]{\ensuremath{C^{\infty}(#1)}}
\newcommand{\singfun}{\ensuremath{C^{\infty}_{\mathrm{Sing}}(0,1)}}
\newcommand{\nonsingfun}{\ensuremath{C^{\infty}_{\mathrm{0}}(0,1)}}
\newcommand{\powser}[1]{\ensuremath{\mathscr{O}(#1)}}
\newcommand{\psloc}[1]{\ensuremath{\mathscr{O}_{\mathrm{loc}}(#1)}}
\newcommand{\asy}[1]{\ensuremath{\mathscr{A}(#1)}}
\newcommand{\cpairs}{\ensuremath{\mathscr{GP}}}
\newcommand{\ctrips}{\ensuremath{\mathscr{GT}}}
\newcommand{\gf}{\ensuremath{\mathbb{R}}}
\newcommand{\ltens}[2]{\ensuremath{#1(\!(#2)\!)}}

\newcommand{\cotimes}{\ensuremath{\hat{\otimes}}}
\newcommand{\noproof}{\begin{flushright} \ensuremath{\square} \end{flushright}}
\newcommand{\fields}{\ensuremath{C^{\infty}(M)}}

\DeclareMathOperator{\Hom}{Hom}
\DeclareMathOperator{\Aut}{Aut}
\DeclareMathOperator{\Bij}{Bij}
\DeclareMathOperator{\Sing}{Sing}

\begin{document}
\title{On two constructions of an effective field theory}
\author{Alastair Hamilton}
\address{Texas Tech University, Department of Mathematics and Statistics, Lubbock TX 79407-1042. USA.}
\email{alastair.hamilton@ttu.edu}
\begin{abstract}
In this paper we discuss two constructions of an effective field theory starting from a local interaction functional. One relies on the well-established graphical combinatorics of the BPHZ algorithm to renormalize divergent Feynman amplitudes. The other, more recent and due to Costello, relies on an inductive construction of local counterterms that uses no graphical combinatorics whatsoever. We show that these two constructions produce the same effective field theory.
\end{abstract}
\thanks{The author gratefully acknowledges the support of the Simons Foundation, grant 279462.}
\maketitle

\section{Introduction}

This article discusses two approaches to constructing effective field theories through the renormalization of a quantum field theory. The first approach, due to Costello \cite{coseffthy}, constructs an effective field theory starting from any choice of a local action functional. It relies on a particular construction of local counterterms which are used to render certain functional integrals defining the effective field theory finite. The algorithm that Costello uses to construct these counterterms is simple and uses no graphical combinatorics. In fact, this algorithm may be formulated without the need to even mention the notion of a Feynman amplitude. One proceeds carefully through the terms in the perturbation series in a specific order, renormalizing each term along the way. An interesting development arising from this perspective on effective field theory lies in the quantization of gauge theories in the Batalin-Vilkovisky formalism, where Costello arrives at the definition of a new and interesting smooth invariant of certain low-dimensional manifolds \cite{costbv}, \cite{coseffthy}.

The second approach we discuss relies on the graphical combinatorics of the well-known algorithm due to Bogoliubov-Parasiuk \cite{bogpar}, Hepp \cite{hepp} and Zimmermann \cite{zimmer} (BPHZ) to replace the divergent Feynman amplitudes in the functional integrals defining the effective field theory with finite values. This algorithm has been a central part of quantum field theory for many decades \cite{collins}. More recently, Connes and Kreimer \cite{ckhopfI}, \cite{ckhopfII} have encoded this algorithm through a Hopf algebra in which the BPHZ algorithm is described in terms of the Birkhoff factorization of characters of this Hopf algebra. Lately, questions have been raised in the literature \cite{costbv}, \cite[\S 1.11]{coseffthy} and elsewhere \cite{review} concerning what, if any, connection may exist between the approach of Costello \cite{coseffthy} and the approach of Connes-Kreimer \cite{ckhopfI} which relies on the BPHZ construction of counterterms for a renormalizable quantum field theory.

In this paper we provide a positive answer to this question. Namely, we show that these two constructions of an effective field theory that come from the work \cite{coseffthy} of Costello and from the BPHZ algorithm coincide. One conclusion of this is that in our situation the BPHZ construction of counterterms has a particularly simple nondiagrammatic formulation, although we should mention that in this case the information regarding the divergence of each graph is lost, as the singular components of a number of graphs are combined into a single counterterm.

Throughout the paper we choose our space of fields $\mathcal{E}$ to be the space of smooth functions on a Riemannian manifold, that is to say we only consider scalar field theories, although it will quickly be evident from the exposition that our results hold for a much broader class of choices for $\mathcal{E}$, such as the space of global sections of a vector bundle of the form described in Section 2.13 of \cite{coseffthy}. However, our choice of $\mathcal{E}$ has the advantage that it's definition is simple, transparent and nontechnical, but yet it retains the appropriate amount of structure required to exhibit the coincidence of the two constructions described above. We mention, of course, that one obvious problem we face in choosing the appropriate level of generality in which to frame our results is that no mathematical consensus currently exists as to what a quantum field theory actually is.

The layout of the paper is as follows. Section \ref{sec_effthy} recalls the basic picture of effective field theory that is outlined by Costello in \cite{coseffthy}. Section \ref{sec_calgo} describes Costello's construction of the local counterterms he uses for renormalizing a quantum field theory and recalls how this gives rise to an effective field theory. Section \ref{sec_BPHZ} gives a careful formulation of the BPHZ algorithm as it applies to our particular situation and proves some basic properties of the counterterms produced. In Section \ref{sec_main} we prove the main theorem of the paper which states that the effective field theory described by Costello and the effective field theory produced by the BPHZ algorithm are the same.

\subsection*{Notation and conventions}

In this paper we will consider topological vector spaces over the real numbers. We will denote the completed projective tensor product of two locally convex topological vector spaces $\mathcal{V}$ and $\mathcal{U}$ by $\mathcal{V}\cotimes \mathcal{U}$; the symbol $\otimes$ will be reserved for the standard tensor product. Since all of the topological vector spaces that we will apply $\cotimes$ to will be nuclear spaces, the projective and injective tensor products coincide. This of course means that we work in a nice symmetric monoidal category.

Given a real topological vector space $\mathcal{V}$, we will denote its \emph{continuous} linear dual by $\mathcal{V}^{\dag}$. Throughout the paper, we will always equip the dual space with the strong topology. We will denote the space of \emph{continuous} linear maps between two topological vector spaces $\mathcal{V}$ and $\mathcal{U}$ by $\Hom(\mathcal{V},\mathcal{U})$. Likewise, this space always carries the strong topology. We will make frequent use of the fact (cf. Proposition 50.7 of \cite{treves}) that if $\mathcal{V}$ and $\mathcal{U}$ are nuclear Fr\'echet spaces then,
\begin{equation} \label{eqn_dualtensor}
\mathcal{V}^{\dag}\cotimes \mathcal{U}^{\dag}=(\mathcal{V}\cotimes \mathcal{U})^{\dag}.
\end{equation}
If $\mathcal{V}$ is a nuclear Fr\'echet space then we denote the algebra of formal power series on $\mathcal{V}$ by
\[ \powser{\mathcal{V}}:=\prod_{n=0}^{\infty} \left[(\mathcal{V}^{\dag})^{\cotimes n}/S_n\right] = \prod_{n=0}^{\infty} \left[\Hom(\mathcal{V}^{\cotimes n},\gf)/S_n\right], \]
where the symmetric group $S_n$ acts in the obvious way by permuting factors.

Given a compact manifold $M$, we denote the space of smooth functions on $M$ by $\smth{M}$. This carries the topology of uniform convergence on compact subsets of the functions and their derivatives. With this topology, $\smth{M}$ is a nuclear Fr\'echet space. We will make frequent use of the result that for two compact manifolds $M$ and $N$,
\[ \smth{M}\cotimes\smth{N} = \smth{M\times N}. \]

If $R$ is a ring then the ring of formal power series in $\hbar$ over $R$ will be denoted by $R[[\hbar]]$. If $A$ and $B$ are finite sets then $|A|$ will refer to the cardinality of $A$ and we will denote the set of bijections between $A$ and $B$ by $\Bij(A,B)$.

\section{Effective field theory} \label{sec_effthy}

In this section we recall some basic background material on effective field theory from \cite{coseffthy}, which follows the work of Wilson \cite{wilson}. We start by fixing some of the geometric and analytic framework that we will work in.

\subsection{Geometric background}

Let $M$ be a compact Riemannian manifold.  We fix our space of fields to be $\mathcal{E}:=\fields$; that is to say that we shall only consider (massive) scalar field theories. We denote the Laplacian on $\fields$ by $\Delta$, whose definition is fixed by the convention that its eigenvalues are nonnegative. The heat kernel
\[ k\in C^{\infty}(M\times M\times (0,\infty)) \]
is defined uniquely by the equation
\[ \int_{y\in M} k (x, y, t)\phi (y) dy = e^{-t(m^2+\Delta)}\phi (x); \qquad  \phi\in\fields,  x\in M, t>0; \]
where $m>0$ is the mass.

From the heat kernel we define a propagator
\[ P(\varepsilon,L)\in\mathcal{E}\cotimes\mathcal{E}=C^{\infty}(M\times M); \qquad \varepsilon>0, 0<L\leq\infty;  \]
by the equation\footnote{If $\varepsilon>L$ we define the propagator by $P(\varepsilon,L):= -P(L,\varepsilon)$.}
\[ P(\varepsilon,L):(x, y)\mapsto \int_{t=\varepsilon}^{t=L} k (x, y, t) dt. \]
Note that since the heat kernel is symmetric in the variables $x$ and $y$, $P(\varepsilon,L)$ is a $\mathbb{Z}_2$-invariant tensor.

Following \cite{coseffthy}, we adopt the following definition of locality.

\begin{defi}
We say a functional $I:\mathcal{E}^{\cotimes k}\to\gf$ is a \emph{local functional} if it is the integral of a product of differential operators; more specifically, we require that $I$ be a sum of operators of the form,
\[ f_1\otimes\cdots\otimes f_k \mapsto \int_M \left[D_1(f_1)\cdots D_k(f_k)\right], \]
where the $D_i:\fields\to\fields$ are differential operators.

We denote the subspace of local functionals by
\[ \Hom_{\mathrm{loc}}(\mathcal{E}^{\cotimes k},\gf). \]
Likewise, we define $\psloc{\mathcal{E}}\subset\powser{\mathcal{E}}$ by
\[ \psloc{\mathcal{E}}:=\prod_{n=0}^\infty\left[\Hom_{\mathrm{loc}}(\mathcal{E}^{\cotimes n},\gf)/S_n\right]. \]
\end{defi}

The kinetic part of our theory will be fixed as
\begin{equation} \label{eqn_kinetic}
K(\phi,\phi):=-\frac{1}{2}\int_{y\in M} \phi(y)\cdot(\Delta+m^2)\phi(y) dy
\end{equation}
and an interaction will be defined as follows.

\begin{defi}
An \emph{interaction} is an element $I\in\powser{\mathcal{E}}[[\hbar]]$. We may write such an interaction as
\[ I=\sum_{i,j=0}^{\infty} \hbar^i I_{ij}, \qquad I_{ij}\in\Hom(\mathcal{E}^{\cotimes j},\gf)/S_j. \]
We require that $I_{00}=I_{01}=I_{02}=0$. If in addition $I\in\psloc{\mathcal{E}}[[\hbar]]$ we say $I$ is a \emph{local interaction}.
\end{defi}

\subsection{Feynman diagram expansion}

The renormalization group flow will be defined, following \cite{coseffthy}, through a Feynman diagram expansion, which in this section we describe rather formally. The reader may wonder why we give such a formal definition of what is after all a very standard part of quantum field theory. The reason is that we wish to use this section to introduce the requisite algebraic and combinatorial background from \cite{cycop}, \cite{modop} and \cite{opalgtophys} that we make use of in the rest of the paper when we come to dealing with the BPHZ algorithm and proving the main theorem.

We begin by recalling from \cite{modop} the definition of a stable graph.

\begin{defi} \label{def_stabgraph}
A stable graph $\Gamma$ consists of a (possibly empty) set $H(\Gamma)$, called the \emph{half-edges} of $\Gamma$, together with the following extra data:
\begin{itemize}
\item
A partition $V(\Gamma)$ of $H(\Gamma)$ into nonempty subsets called the \emph{vertices} of $\Gamma$.
\item
A labeling of each vertex $v\in V(\Gamma)$ by a nonnegative integer $g_v$ called the \emph{genus}. We insist that vertices of genus zero are at least trivalent.
\item
A subset $L(\Gamma)$ of $H(\Gamma)$ called the \emph{legs} of $\Gamma$.
\item
A partition $E(\Gamma)$ of the remaining half-edges $H(\Gamma)-L(\Gamma)$ into pairs called the \emph{edges} of $\Gamma$.
\end{itemize}
The \emph{genus} $g(\Gamma)$ of a stable graph $\Gamma$ is defined by
\[ g(\Gamma):= \dim(H_1(\Gamma,\mathbb{Q})) + \sum_{v\in V(\Gamma)} g_v. \]
\end{defi}

We now recall how to label the factors in a tensor product by a set.

\begin{defi}
Let $\mathcal{V}$ be a nuclear space and $A$ be a finite set of cardinality $n:=|A|$. Define
\[ \ltens{\mathcal{V}}{A}:= \left[\bigoplus_{f\in\Bij(\{1,\ldots, n\},A)} \mathcal{V}^{\cotimes n}\right]^{S_n}; \]
where we consider the $S_n$-invariants of the action that simultaneously permutes the summands and the tensor factors.
\end{defi}

There is an obvious family of isomorphisms of topological vector spaces,
\begin{equation} \label{eqn_tensorin}
i_f: \mathcal{V}^{\cotimes n}\to \ltens{\mathcal{V}}{A}, \quad f\in\Bij(\{1,\ldots, n\},A);
\end{equation}
satisfying $i_{f\sigma}=i_f\circ\sigma$ for all $\sigma\in S_n$. Consequently, there is a well-defined canonical map,
\begin{equation} \label{eqn_symquot}
\ltens{\mathcal{V}}{A} \to \mathcal{V}^{\cotimes n}/S_n,
\end{equation}
that does not depend on picking a bijection. As noted in \cite{cycop} and \cite[\S II.1.7]{opalgtophys}, $\ltens{\mathcal{V}}{\cdot}$ is a functor on the category of sets where morphisms are bijections. We can also make the canonical identification,
\[ \ltens{\mathcal{V}}{A\sqcup B} = \ltens{\mathcal{V}}{A}\cotimes\ltens{\mathcal{V}}{B}, \]
and provided that $\mathcal{V}$ is also a Fr\'echet space,
\[ \ltens{\mathcal{V}^{\dag}}{A} = {\ltens{\mathcal{V}}{A}}^{\dag}; \]
where here we have used \eqref{eqn_dualtensor}.

We now place the well-known definition of Feynman amplitude in this framework.

\begin{defi} \label{def_feynamp}
Let $I\in\powser{\mathcal{E}}[[\hbar]]$ be an interaction and take a propagator $P\in\mathcal{E}\cotimes\mathcal{E}$ (by a propagator, we just mean a $\mathbb{Z}_2$-invariant tensor). Define
\begin{equation} \label{eqn_itilde}
\widetilde{I}_{ij}:= \sum_{\sigma\in S_j} \sigma\cdot I_{ij} \in \big(\mathcal{E}^{\dag}\big)^{\cotimes j}.
\end{equation}
Given a stable graph $\Gamma$, we define its Feynman amplitude as follows. Decorate each vertex $v\in V(\Gamma)$ of genus $g_v$ by $\widetilde{I}_{g_v,|v|}\in\ltens{\mathcal{E}^{\dag}}{v}$ using \eqref{eqn_tensorin}. Combining all these yields a tensor,
\begin{equation} \label{eqn_feynint}
\bigotimes_{v\in V(\Gamma)} \left[ \widetilde{I}_{g_v,|v|} \right] \in\mathcal{E}^{\dag}\bigg(\!\!\bigg(\bigcup_{v\in V(\Gamma)}v\bigg)\!\!\bigg) = \ltens{\mathcal{E}^{\dag}}{H(\Gamma)}.
\end{equation}
Now, using \eqref{eqn_tensorin} again, decorate every edge $e\in E(\Gamma)$ with the propagator $P\in\ltens{\mathcal{E}}{e}$. Combining these yields a tensor,
\begin{equation} \label{eqn_feynprop}
\bigotimes_{e\in E(\Gamma)} \left[ P \right]\in \mathcal{E}\bigg(\!\!\bigg( \bigcup_{e\in E(\Gamma)} e \bigg)\!\!\bigg).
\end{equation}
Evaluating the tensor \eqref{eqn_feynint} upon the tensor \eqref{eqn_feynprop} yields a tensor
\begin{equation} \label{eqn_feynamp}
F_{\Gamma}(I,P) \in \ltens{\mathcal{E}^{\dag}}{L(\Gamma)},
\end{equation}
called the Feynman amplitude of $\Gamma$. Applying the map \eqref{eqn_symquot} to \eqref{eqn_feynamp} yields a monomial of degree $|L(\Gamma)|$ in \powser{\mathcal{E}}. We will frequently abuse notation and denote this monomial by $F_{\Gamma}(I,P)$ as well. The context will make it clear which meaning is intended, as in the following definition.
\end{defi}

\begin{defi}
Let $I\in\powser{\mathcal{E}}[[\hbar]]$ be an interaction and $P\in\mathcal{E}\cotimes\mathcal{E}$ be a propagator, as above. We define the Feynman diagram expansion by
\[ W(I,P):= \sum_{i=0}^{\infty} \hbar^iI_{i0} + \sum_{\Gamma} \frac{\hbar^{g(\Gamma)}}{|\Aut(\Gamma)|} F_{\Gamma}(I,P)\in\powser{\mathcal{E}}[[\hbar]], \]
where we sum over all (isomorphism classes of) connected\footnote{Note that the empty set is not considered to be a connected graph.} stable graphs $\Gamma$.
\end{defi}

\begin{rem}
The above expression could be defined without using Feynman diagrams as
\begin{equation} \label{eqn_feynographs}
W(I,P) = \hbar\ln\left[\exp(\hbar\partial_P)[\exp(I/\hbar)]\right],
\end{equation}
where $\partial_P$ is the second order differential operator on $\powser{\mathcal{E}}$ associated to the propagator $P$, cf. \cite[\S 2.3.4]{coseffthy}. This means that everything we define in this section could have been introduced without any need to ever mention Feynman diagrams. We point this out not because we wish to adopt this slightly unconventional point of view. Indeed, we need to work with Feynman diagrams so that we may eventually make contact with the BPHZ algorithm in sections \ref{sec_BPHZ} and \ref{sec_main}. Instead, we wish to emphasize that Costello's algorithm and construction of an effective field theory does not rely upon any graphical combinatorics and could in fact be formulated without ever using Feynman diagrams.
\end{rem}

Let us denote the homogeneous part of $W(I,P)$ of order $i$ in $\hbar$ and $j$ in $\mathcal{E}$ by
\[ W_{ij}(I,P)\in\Hom(\mathcal{E}^{\cotimes j},\mathbb{R})/S_j. \]
In this way we may write
\[ W(I,P) = \sum_{i,j=0}^{\infty} \hbar^i W_{ij}(I,P). \]
A convenient formula for $W_{ij}(I,P)$ is,
\begin{equation} \label{eqn_feycomp}
W_{ij}(I,P) = I_{ij} \ + \ \sum_{\Gamma: \ \begin{subarray}{c} g(\Gamma)=i, \ |L(\Gamma)|=j, \\ |E(\Gamma)|>0. \end{subarray}}\frac{1}{|\Aut(\Gamma)|}F_{\Gamma}(I,P); \qquad i,j\geq 0.
\end{equation}
Here we sum over all connected stable graphs with at least one edge and having both a fixed genus and a fixed number of legs.

Following \cite{coseffthy} we define a well-ordering on the indexing set as follows,
\begin{equation} \label{eqn_indexorder}
\left[(i,j)<(i',j')\right]\Leftrightarrow \left[i<i'\text{ or }(i=i'\text{ and }j<j')\right];\qquad i,j\geq 0.
\end{equation}

\begin{lemma} \label{lem_basicid}
Let $I\in\powser{\mathcal{E}}[[\hbar]]$ be an interaction and $P\in\mathcal{E}\cotimes\mathcal{E}$ be a propagator as before and suppose that
\[ J:=\sum_{(p,q)\geq (i,j)} \hbar^p J_{pq}, \qquad J_{pq}\in\Hom(\mathcal{E}^{\cotimes q},\gf)/S_q; \]
is a power series consisting of terms of order $\geq(i,j)$. Then
\[ W(I-J,P) = W(I,P) - \hbar^iJ_{ij} + \text{terms of order}>(i,j). \]
\end{lemma}

\begin{proof}
This follows from a very simple calculation.
\end{proof}

\subsection{Effective field theory}

In this section we recall from \cite{coseffthy} Costello's definition of an effective field theory. We will work in the length scale formulation rather than the more conceptual energy scale formulation, which means that our functional integrals are regularized by cutting off all the contributions that come from allowing a particle to propagate for a proper time less than some small $\varepsilon$ before interacting. The reason for this preference, as Costello carefully explains in \cite{coseffthy}, is that it is much easier to incorporate locality in this picture. We also have another motive in mind,  which is that we would like to make use of Costello's results regarding the asymptotic behavior of the length scale regularization of these functional integrals, which involves a fairly lengthy and technical analysis of the small $t$ asymptotic expansion of the heat kernel \cite[\S A1.5]{coseffthy}. We mention that Costello shows how one may pass freely between the length and energy scale formulations \cite[\S 2.12]{coseffthy}.

\begin{defi} \label{def_effthy}
An effective (scalar) field theory with kinetic term \eqref{eqn_kinetic} is a family of effective interactions,
\[ I[L]\in\powser{\mathcal{E}}[[\hbar]], \qquad 0<L\leq\infty; \]
satisfying the following two conditions:
\begin{enumerate}
\item
The renormalization group equation,
\[ I[L_2]=W\left(I[L_1],P(L_1,L_2)\right); \qquad 0<L_1\leq L_2\leq\infty. \]
\item
The asymptotic locality requirement on this family that there be a small $L$ asymptotic expansion,
\[ I_{ij}[L]\simeq\sum_{k=0}^\infty g_k(L)\Phi_k, \qquad i,j\geq 0; \]
where $g_k\in C^{\infty}(0,\infty)$ and $\Phi_k\in\Hom_{\mathrm{loc}}(\mathcal{E}^{\cotimes j},\gf)/S_j$.
\end{enumerate}
\end{defi}

\begin{rem}
For the sake of clarity, we recall from \cite{coseffthy} the precise meaning of asymptotic expansion in this instance. This means that there is a nondecreasing sequence $d_n\in\mathbb{Z}$, tending to infinity, such that for all $n\geq 0$,
\begin{equation} \label{eqn_asymptotic}
\lim_{L\to 0} L^{-d_n}\left(I_{ij}[L]-\sum_{k=0}^n g_k(L)\Phi_k\right)=0.
\end{equation}
We note here that it makes no difference in \eqref{eqn_asymptotic} whether we equip $\Hom(\mathcal{E}^{\cotimes j},\gf)$ with the weak or the strong topology. This is because $\mathcal{E}$ is a nuclear Fr\'echet space and hence a Montel space by Corollary 3 of Proposition 50.2 in \cite{treves}. For Montel spaces, convergence in the weak and strong topologies are equivalent, cf. Corollary 1 of Proposition 34.6 in \cite{treves}. In fact, although \cite{treves} only states this result for the dual of a Montel space, a careful examination of the argument provided in \cite{treves} shows that this statement remains true when we replace $\gf$ by any locally convex topological vector space.
\end{rem}

\section{Costello's algorithm} \label{sec_calgo}

In this section we describe Costello's algorithm \cite{coseffthy} for producing counterterms for a quantum field theory and the concomitant construction of an effective field theory. As we will shortly see, this inductive algorithm uses no graphical combinatorics and is in fact very simple to formulate.

\subsection{Singular component of an asymptotic expansion}

We begin by recalling from \cite{coseffthy} the method for extracting the singular part of the Feynman amplitude $F_{\Gamma}(I,P(\varepsilon,L))$, which makes use of Costello's analysis \cite[\S A1.5]{coseffthy} of their short length asymptotic behavior. We start with the usual definition (cf. \cite[\S 2.9]{coseffthy}) of a renormalization scheme.

\begin{defi}
A renormalization scheme is a choice of decomposition,
\begin{equation} \label{eqn_rscheme}
\smth{0,1}=\singfun\oplus\nonsingfun,
\end{equation}
of the space of smooth functions on the open unit interval into a direct sum of the space $\nonsingfun$ consisting of those functions $f (\varepsilon)$ admitting a limit as $\varepsilon\to 0$, and a complimentary subspace $\singfun$ of `purely singular' functions. We will denote the operator of projection onto the singular part by
\[ T:\smth{0,1}\to\singfun. \]
\end{defi}

From now and for the remainder of the paper we fix a choice \eqref{eqn_rscheme} of renormalization scheme. In particular, we use the same renormalization scheme in this section and Section \ref{sec_BPHZ}.

\begin{defi}
Let $\mathcal{W}$ be a locally convex Hausdorff topological vector space and consider the subspace $\asy{\mathcal{W}}$ of the space of functions from $(0,1)$ to $\mathcal{W}$, which consists of those functions $f$ that have a small $\varepsilon$ asymptotic expansion of the form,
\begin{equation} \label{eqn_smlepsasym}
f(\varepsilon)\simeq\sum_{i=0}^{\infty} g_i(\varepsilon)\Psi_i,
\end{equation}
where $\Psi_i\in\mathcal{W}$ and the $g_i\in\smth{0,1}$ have finite order poles at zero. More precisely, this means that there is a nondecreasing sequence $d_n\in\mathbb{Z}$, tending to infinity, such that for all $n\geq 0$,
\[ \lim_{\varepsilon\to 0}\varepsilon^{-d_n}\left(f(\varepsilon) - \sum_{i=0}^{n} g_i(\varepsilon)\Psi_i\right) = 0. \]
$\asy{\cdot}$ is a functor from the category of locally convex Hausdorff topological vector spaces to the category of vector spaces.
\end{defi}

For functions having asymptotic expansions of the form \eqref{eqn_smlepsasym}, we may define their singular part using our renormalization scheme \eqref{eqn_rscheme} as follows.

\begin{defi}
Suppose that $f\in\asy{\mathcal{W}}$ has small $\varepsilon$ asymptotic expansion \eqref{eqn_smlepsasym}. Then there exists $N\in\mathbb{N}$ such that $g_n(\varepsilon)\to 0$ as $\varepsilon\to 0$, for all $n\geq N$. Define,
\[ \Sing(f):=\sum_{i=0}^N T(g_i) \Psi_i. \]
Elementary arguments show that $\Sing(f)$ does not depend on $N$, or the form \eqref{eqn_smlepsasym} of the asymptotic expansion that is chosen for $f$.
\end{defi}

We collect some basic facts about the operator $\Sing$:
\begin{itemize}
\item
$\Sing:\asy{\mathcal{W}}\to\asy{\mathcal{W}}$ is an idempotent linear operator. It is natural in $\mathcal{W}$.
\item
Given a function $f\in\asy{\mathcal{W}}$, the limit $\lim_{\varepsilon\to 0} f(\varepsilon)$ exists if and only if $\Sing(f)=0$.
\item
Given $f\in\asy{\mathcal{W}}$, the limit $\lim_{\varepsilon\to 0}[f-\Sing(f)](\varepsilon)$ always exists.
\end{itemize}

The proofs of these facts are routine. In what follows we will make extensive use of the following difficult result of Costello, cf. Theorem 9.3.1 in \S 2.9 and Theorem 4.0.2 in \S A1.4 of \cite{coseffthy}.

\begin{theorem} \label{thm_feynasy}
Let $I\in\psloc{\mathcal{E}}[[\hbar]]$ be a local interaction and let $\Gamma$ be a connected stable graph. Consider the function,
\begin{equation} \label{eqn_epsilonmap}
\begin{array}{ccc}
(0,1) & \to & \Hom(\ltens{\mathcal{E}}{L(\Gamma)},\smth{0,\infty}), \\
\varepsilon & \mapsto & \left[\mathbf{a}\mapsto (L\mapsto F_{\Gamma}(I,P(\varepsilon,L))[\mathbf{a}]) \right];
\end{array}
\end{equation}
which we denote by $F_{\Gamma}(I,P(-,-))$. Then,
\begin{enumerate}
\item \label{item_feynasyexp}
$F_{\Gamma}(I,P(-,-))\in\asy{\Hom(\ltens{\mathcal{E}}{L(\Gamma)},\smth{0,\infty})}$, that is there is a small $\varepsilon$ asymptotic expansion,
\[ F_{\Gamma}(I,P(\varepsilon,-))\simeq\sum_{i=0}^{\infty} g_i(\varepsilon)\Psi_i, \]
as in \eqref{eqn_smlepsasym}.
\item \label{item_feynasyloc}
Moreover, each $\Psi_i\in\Hom(\ltens{\mathcal{E}}{L(\Gamma)},\smth{0,\infty})$ has a small $L$ asymptotic expansion in terms of local action functionals,
\[ \Psi_i(L)\simeq\sum_{j=0}^\infty f_{ij}(L)\psi_{ij}, \]
meaning each $\psi_{ij}\in\Hom_{\mathrm{loc}}(\ltens{\mathcal{E}}{L(\Gamma)},\gf)$ and each $f_{ij}\in\smth{0,\infty}$.
\end{enumerate}
\end{theorem}
\noproof

\subsection{Construction of an effective field theory} \label{sec_effthyconst}

We now recall Costello's formula \cite[\S 2.10]{coseffthy} for producing local counterterms for a local interaction and his construction of an effective field theory that naturally follows from it. We begin with Theorem 10.1.1 of \cite[\S 2.10]{coseffthy}, whose proof involves the repeated use of Lemma \ref{lem_basicid}.

\begin{theorem} \label{thm_counterterms}
Let $I\in\psloc{\mathcal{E}}[[\hbar]]$ be a local interaction, then there is a series of local counterterms,
\[ \Big[ I_{ij}^{CT}:\varepsilon\mapsto I_{ij}^{CT} (\varepsilon)\Big]\in\singfun\otimes\Hom_{\mathrm{loc}}\big(\mathcal{E}^{\cotimes j},\gf\big)/S_j, \]
satisfying,
\begin{equation} \label{eqn_counterterms}
I_{ij}^{CT} = \Sing\left[\varepsilon\mapsto W_{ij}\left(I-\sum_{(p,q)<(i, j)} \hbar^p I_{pq}^{CT}(\varepsilon),P(\varepsilon,L)\right)\right]
\end{equation}
for all $i, j\geq 0$ and such that the limit,
\[ \lim_{\varepsilon\to 0}\left[W\left(I-\sum_{i,j=0}^{\infty} \hbar^i I_{ij}^{CT}(\varepsilon),P(\varepsilon, L)\right)\right] \]
exists in $\powser{\mathcal{E}}[[\hbar]]$ for all $0 <L\leq\infty$.
\end{theorem}
\noproof

\begin{rem}
Equation \eqref{eqn_counterterms} may be taken as an inductive definition of the counterterms $I_{ij}^{CT}$, where we remind the reader that we have used the well-ordering on the index set defined by \eqref{eqn_indexorder}. According to this equation, $I_{ij}^{CT}(\varepsilon)$ should lie in the space $\Hom(\mathcal{E}^{\cotimes j},\smth{0,\infty})/S_j$, due to this expression's dependence on $L$. Part of Costello's proof of Theorem \ref{thm_counterterms} involves showing that this expression does not in fact depend on $L$, that is $I_{ij}^{CT}(\varepsilon)$ lies in $\Hom(\mathcal{E}^{\cotimes j},\gf)/S_j$.
\end{rem}

\begin{rem}
We can say a little more about these counterterms in fact. The counterterms $I_{ij}^{CT}$ are finite sums of terms of the form $g\cdot\psi$ where $\psi$ is a local functional and $g$ is a purely singular function \emph{with a finite order pole at zero}. This ensures, by Theorem \ref{thm_feynasy}, that the expression in Equation \eqref{eqn_counterterms} has an asymptotic expansion and hence that its singular part is well-defined.
\end{rem}

The construction of the counterterms \eqref{eqn_counterterms} leads naturally to the definition of an effective field theory satisfying all the requirements of Definition \ref{def_effthy}. This is described in \cite[\S 2.11]{coseffthy} and summarized by the following theorem.

\begin{theorem} \label{thm_effthy}
Let $I\in\psloc{\mathcal{E}}[[\hbar]]$ be a local interaction and define a family of effective interactions by
\begin{equation} \label{eqn_effthy}
I[L]:=\lim_{\varepsilon\to 0}\left[ W\left( I-\sum_{i,j=0}^\infty \hbar^iI_{ij}^{CT}(\varepsilon),P(\varepsilon,L) \right) \right].
\end{equation}
Then this family of interactions forms an effective field theory satisfying the axioms of Definition \ref{def_effthy}.
\end{theorem}
\noproof

\section{Graphical Combinatorics and the BPHZ algorithm} \label{sec_BPHZ}

In this section we give a precise description of the BPHZ algorithm \cite{bogpar}, \cite{hepp}, \cite{zimmer} as it applies to our situation, following closely the description provided by Collins \cite[\S 5]{collins}, and prove some basic properties of the counterterms that it produces. Typically, the BPHZ algorithm is applied when working in a momentum space formulation. Here, we take the slightly unconventional approach of applying it to our position space formulation. This decision will be entirely justified when we later prove in Theorem \ref{thm_main} that we recover Costello's construction \eqref{eqn_effthy} of an effective field theory in this way.

\subsection{Basic definitions}

We start by collecting some basic definitions concerning operations on graphs. The operations that we introduce on graphs, namely contracting and inserting subgraphs, are standard operations in the theory of operads, cf. \cite{opalgtophys}. They were also used in the descriptions \cite{ckhopfI}, \cite{ckinsert} of the BPHZ algorithm through Hopf algebras. We begin by recalling from \cite{modop} how to contract edges and loops in a stable graph.

\begin{defi}
Suppose that $\Gamma$ is a stable graph and $e\in E(\Gamma)$ is an edge.
\begin{enumerate}
\item
If $e$ is a loop then $\Gamma/e$ is the stable graph that results by throwing $e$ away and increasing the genus of the incident vertex by $1$.
\item
If $e$ is not a loop then we form $\Gamma/e$ by contracting the edge $e$ and coalescing the two incident vertices into a single vertex, whose genus is the sum of the genera of the two incident vertices.
\end{enumerate}
\end{defi}

\begin{rem}
Note that $\Gamma/e$ is not defined if this would force the new vertex that is formed by contracting the edge or loop $e$ to be empty.
\end{rem}

We make the following definition of a subgraph.

\begin{defi}
A subgraph $\gamma$ of a stable graph $\Gamma$ is just a subset of the edges $E(\Gamma)$ of $\Gamma$. We say that $\gamma$ is a \emph{proper subgraph} if it is a proper subset of $E(\Gamma)$. If every connected component of $\Gamma$ has a nonempty set of legs, or if $\Gamma$ is connected and $\gamma$ is a proper subgraph, then since it does not matter which order we contract edges in, we may define $\Gamma/\gamma$ to be the stable graph obtained by contracting all the edges of $\gamma$, see Figure \ref{fig_contractgraph}.

To any subgraph $\gamma$ of $\Gamma$ we may associate an actual stable graph in the sense of Definition \ref{def_stabgraph}, which is pictured on the left of Figure \ref{fig_contractgraph}. By an abuse of notation, we will denote this stable graph by the same symbol $\gamma$. It is defined as follows:
\begin{itemize}
\item
By definition, the subgraph $\gamma$ specifies a subset of the edges of $E(\Gamma)$. We define $E(\gamma)$ to be this subset.
\item
The vertices of $\gamma$ consist of all those vertices of $\Gamma$ intersecting the subgraph,
\[ V(\gamma):=\Big\{v\in V(\Gamma): v\cap\Big(\underset{e\in E(\gamma)}{\cup} e\Big)\neq\emptyset \Big\}. \]
These vertices have the same genus as those of the original graph $\Gamma$.
\item
This determines the half-edges and legs of $\gamma$;
\[ H(\gamma):=\underset{v\in V(\gamma)}{\cup} v, \qquad L(\gamma):=H(\gamma)-\underset{e\in E(\gamma)}{\cup} e. \]
\end{itemize}
\end{defi}

\begin{figure}[h]
\includegraphics{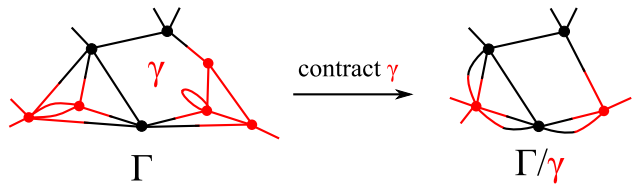}
\caption{Contracting the red subgraph $\gamma$ yields the graph on the right. The legs of the components of $\gamma$ form vertices in this new graph.}
\label{fig_contractgraph}
\end{figure}

Suppose that $\Gamma$ is a stable graph and that $\gamma=\cup_{i\in\mathcal{I}}\gamma_i$ is a subgraph of $\Gamma$ that we have written as a union of its connected components $\gamma_i$. Provided we are allowed to contract the subgraph $\gamma$, it is easy to see (and not hard to prove) that:
\begin{itemize}
\item
the legs of each connected component $\gamma_i$ form a vertex of $\Gamma/\gamma$;
\begin{equation} \label{eqn_legcontract}
\begin{array}{ccc}
\text{Components of } \gamma & \to & V(\Gamma/\gamma), \\
\gamma_i & \mapsto & v_{\gamma_i}:=L(\gamma_i);
\end{array}
\end{equation}
\item
that the image $\mathscr{V}_{\gamma}\subset V(\Gamma/\gamma)$ of the map \eqref{eqn_legcontract} satisfies,
\begin{equation} \label{eqn_vertexcontract}
V(\Gamma/\gamma) - \mathscr{V}_{\gamma} = V(\Gamma)-V(\gamma);
\end{equation}
\item
that the map \eqref{eqn_legcontract} respects genera, i.e. that the genus of a vertex formed by contracting a connected subgraph is the genus of that subgraph,
\[ g^{\Gamma/\gamma}(v_{\gamma_i}) = g(\gamma_i), \quad i\in\mathcal{I}; \]
\item
that the identity \eqref{eqn_vertexcontract} also respects genera, i.e. that the genus of a vertex that survives in $\Gamma/\gamma$ is the same as in the original graph $\Gamma$,
\[ g^{\Gamma}(v) = g^{\Gamma/\gamma}(v), \quad \text{for all } v\in V(\Gamma/\gamma)-\mathscr{V}_{\gamma}. \]
\end{itemize}

These facts may be observed in Figure \ref{fig_contractgraph}. The inverse operation to contracting a subgraph is inserting a graph.

\begin{defi} \label{def_graphinsert}
Suppose that $\Gamma$ and $\gamma=\cup_{i\in\mathcal{I}}\gamma_i$ are stable graphs, where $\Gamma$ is connected and $\gamma$ is written as a union of its connected components $\gamma_i$, each of which have nonempty sets of legs and edges. Suppose further that there is an injective map $\tau: L(\gamma)\to H(\Gamma)$ satisfying:
\begin{enumerate}
\item \label{item_legstovert}
$\tau$ maps the legs of connected components to vertices,
\[ \tau(L(\gamma_i))\in V(\Gamma), \quad\text{for all }i\in\mathcal{I} ;\]
\item \label{item_genusglue}
that this mapping respects the genus,
\[ g^{\Gamma}(\tau(L(\gamma_i)))=g(\gamma_i), \quad\text{for all }i\in\mathcal{I}. \]
\end{enumerate}
From this we may define\footnote{We spare the reader the precise details of a formal construction in terms of edges, vertices and so forth, which are straightforward anyway, since any two such formal constructions must after all produce naturally isomorphic functors in Theorem \ref{thm_eqvcat}.} a connected stable graph $\Gamma\circ_{\tau}\gamma$ by simply inserting the graph $\gamma$ inside $\Gamma$ using the map $\tau$.

It follows from condition \eqref{item_genusglue} that $g(\Gamma\circ_{\tau}\gamma)=g(\Gamma)$. There is a natural subgraph of $\Gamma\circ_{\tau}\gamma$ defined by the edges of $\gamma$. By an abuse of notation, we will denote this subgraph by the same symbol $\gamma$. This is justified as the stable graph that is associated to this subgraph is isomorphic to $\gamma$.
\end{defi}

Consider the following two categories. The objects of the first category $\cpairs$ are pairs $(\Gamma,\gamma)$, where $\Gamma$ is a connected stable graph and $\gamma$ is a proper subgraph of $\Gamma$. A morphism in this category is an isomorphism of stable graphs that preserves the proper subgraphs.

The objects of the second category $\ctrips$ are triples $(\Gamma,\tau,\gamma)$; where $\Gamma$, $\tau$ and $\gamma$ satisfy the requirements of Definition \ref{def_graphinsert}, plus the additional requirement that the edges of $\Gamma$ are nonempty. Morphisms in this category consist of a pair of isomorphisms between stable graphs which commute with the maps defined by the $\tau$.

The statement that the contraction operation is inverse to insertion is justified by the following theorem (which is also much more precise).

\begin{theorem} \label{thm_eqvcat}
There is an equivalence of categories;
\begin{displaymath}
\begin{array}{ccc}
\ctrips & \rightleftharpoons & \cpairs, \\
(\Gamma,\tau,\gamma) & \mapsto & (\Gamma\circ_{\tau}\gamma,\gamma), \\
(\Gamma/\gamma,i,\gamma) & \mapsfrom & (\Gamma,\gamma);
\end{array}
\end{displaymath}
where $i$ denotes the obvious inclusion map.
\end{theorem}

\begin{proof}
The proof of this theorem is a simple check using the observations that we made earlier in this section. Starting on the right, the composition of the two functors is the identity. Starting on the left, an isomorphism from the identity functor to the composition may be constructed from the map $\tau$.
\end{proof}

Next, we will introduce a definition that generalizes the definition of the Feynman amplitude from Definition \ref{def_feynamp}. This will explain how to define a Feynman amplitude in which we replace a subgraph of the stable graph with a different expression.

\begin{defi}
A functional $\nu$ on stable graphs is a mapping that assigns to every stable graph $\Gamma$, an element
\[ \nu_{\Gamma}\in\ltens{\mathcal{E}^{\dag}}{L(\Gamma)}. \]
We will say that a functional $\nu$ is \emph{equivariant} if for every isomorphism $\phi:\Gamma\to\Gamma'$ of stable graphs we have,
\[ \nu_{\Gamma'} = \phi^{\#}(\nu_{\Gamma}). \]
A functional $\nu$ will be called \emph{multiplicative} if
\[ \nu_{\Gamma\sqcup\Gamma'} = \nu_{\Gamma}\otimes\nu_{\Gamma'}. \]
\end{defi}

\begin{rem}
An example of an equivariant multiplicative functional on stable graphs is provided by the usual Feynman amplitude from Definition \ref{def_feynamp}.
\end{rem}

\begin{defi} \label{def_subfamp}
Suppose that $\nu$ is a functional on stable graphs, $I\in\powser{\mathcal{E}}[[\hbar]]$ is an interaction and $P\in\mathcal{E}\cotimes\mathcal{E}$ is a propagator. Suppose that $\Gamma$ is a stable graph and that $\gamma$ is a subgraph of $\Gamma$ for which $\Gamma/\gamma$ is duly defined. The functional $\nu$ defines a tensor
\[ \nu_{\gamma}\in\ltens{\mathcal{E}^{\dag}}{L(\gamma)}=\mathcal{E}^{\dag}\bigg(\!\!\bigg( \bigcup_{v\in\mathscr{V}_{\gamma}}v \bigg)\!\!\bigg), \]
where $\mathscr{V}_{\gamma}\subset V(\Gamma/\gamma)$ was defined by \eqref{eqn_legcontract}. Combining this with the tensor
\[ \omega:=\bigotimes_{v\in V(\Gamma/\gamma)-\mathscr{V}_{\gamma}} \left[\widetilde{I}_{g_v,|v|}\right]\in \mathcal{E}^{\dag}\bigg(\!\!\bigg( \bigcup_{v\in V(\Gamma/\gamma)-\mathscr{V}_{\gamma}}v \bigg)\!\!\bigg), \]
where $\widetilde{I}_{ij}$ was defined by \eqref{eqn_itilde}, yields a tensor $\nu_{\gamma}\otimes\omega\in\ltens{\mathcal{E}^{\dag}}{H(\Gamma/\gamma)}$. Evaluating this tensor on the tensor
\[ \bigotimes_{e\in E(\Gamma/\gamma)} \left[ P \right]\in \mathcal{E}\bigg(\!\!\bigg( \bigcup_{e\in E(\Gamma/\gamma)} e \bigg)\!\!\bigg) \]
yields an element,
\[ F_{(\Gamma,\gamma;\nu)}(I,P)\in\ltens{\mathcal{E}^{\dag}}{L(\Gamma/\gamma)}=\ltens{\mathcal{E}^{\dag}}{L(\Gamma)}. \]
This defines the Feynman amplitude of $\Gamma$ with subgraph $\gamma$ replaced by $\nu$.
\end{defi}

\begin{rem}
When $\gamma$ is the empty subgraph, this is just the usual Feynman amplitude \eqref{eqn_feynamp}, providing $\nu_{\emptyset}=1$.
\end{rem}

\subsection{The BPHZ algorithm}

We now use Definition \ref{def_subfamp} to give a precise formulation of the BPHZ algorithm in our framework, following closely the description provided in Collins' textbook \cite[\S 5]{collins} which was used by Connes and Kreimer in \cite{ckhopfI}. We mention that the cited sources work in a momentum space representation, whilst we work in position space. Notwithstanding this detail, we will see that the combinatorics involved are the same.

The counterterms produced will be parameter dependent multiplicative functionals on stable graphs,
\[ C_*(\varepsilon,L):\Gamma\mapsto C_{\Gamma}(\varepsilon,L)\in\ltens{\mathcal{E}^{\dag}}{L(\Gamma)};\qquad \varepsilon\in (0,1), L\in (0,\infty). \]
As in \eqref{eqn_epsilonmap}, we may consider these counterterms as maps,
\begin{displaymath}
\begin{array}{ccc}
(0,1) & \to & \Hom(\ltens{\mathcal{E}}{L(\Gamma)},\smth{0,\infty}), \\
\varepsilon & \mapsto & \left[\mathbf{a}\mapsto (L\mapsto C_{\Gamma}(\varepsilon,L)[\mathbf{a}]) \right];
\end{array}
\end{displaymath}
which we will denote by $C_{\Gamma}$. This statement of course implies that the functional $C_*(\varepsilon,L)$ will vary smoothly with the parameter $L$.

Since $C_*(\varepsilon, L)$ is a multiplicative functional, it suffices to define it on connected stable graphs.

\begin{defi} \label{def_BPHZ}
Let $I\in\psloc{\mathcal{E}}[[\hbar]]$ be a local interaction. If $\Gamma$ is a connected stable graph, then $C_{\Gamma}(\varepsilon,L)$ is defined inductively on the number of edges of $\Gamma$ so that it satisfies the defining equation,
\[ C_{\Gamma} = -\Sing\left[ F_{\Gamma}(I,P(\varepsilon,L)) + \sum_{\emptyset\varsubsetneq\gamma\varsubsetneq\Gamma} F_{(\Gamma,\gamma;C_*(\varepsilon,L))}(I,P(\varepsilon,L)) \right]; \]
where the sum is taken over all nonempty proper subgraphs $\gamma$ of $\Gamma$ and we consider the terms inside the square brackets as elements of $\asy{\Hom(\ltens{\mathcal{E}}{L(\Gamma)},\smth{0,\infty})}$.

To ease notation, it is customary to define the expression inside the parentheses as,
\[ \overline{R}_{\Gamma}(\varepsilon,L):= F_{\Gamma}(I,P(\varepsilon,L)) + \sum_{\emptyset\varsubsetneq\gamma\varsubsetneq\Gamma} F_{(\Gamma,\gamma;C_*(\varepsilon,L))}(I,P(\varepsilon,L))\in\ltens{\mathcal{E}^{\dag}}{L(\Gamma)}. \]
We denote the corresponding element of $\asy{\Hom(\ltens{\mathcal{E}}{L(\Gamma)},\smth{0,\infty})}$ by $\overline{R}_{\Gamma}$. Then the defining equation becomes simply,
\begin{equation} \label{eqn_BPHZdef}
C_{\Gamma} = -\Sing\left[\overline{R}_{\Gamma}\right].
\end{equation}
\end{defi}

Technically, our definition of the counterterms $C_{\Gamma}$ is not complete because we have not explained why $\overline{R}_{\Gamma}$ has an asymptotic expansion. This will follow from Theorem \ref{thm_feynasy} and the following lemma. In particular, the counterterms $C_{\Gamma}(\varepsilon,L)$ should be \emph{local} in order to satisfy the hypothesis of Theorem \ref{thm_feynasy}. Of course, one of the lauded features of the BPHZ algorithm is its propensity to produce local counterterms, so this part of the lemma will come as no surprise.

\begin{lemma}
Let $I\in\psloc{\mathcal{E}}[[\hbar]]$ be a local interaction.
\begin{enumerate}
\item \label{item_BPHZnonsing}
If $\Gamma$ is a stable graph such that either $\Gamma$ is connected or every connected component of $\Gamma$ has a nonempty set of legs and edges then,
\[ \Sing(\overline{R}_{\Gamma}+C_{\Gamma}) = 0; \]
hence the expression $\overline{R}_{\Gamma}(\varepsilon, L)+C_{\Gamma}(\varepsilon, L)$ converges as $\varepsilon\to 0$.
\item \label{item_BPHZlind}
The counterterms $C_{\Gamma}(\varepsilon,L)$ do not depend on the parameter $L$, that is they are functions,
\[ C_{\Gamma}: (0,1)\to\Hom(\ltens{\mathcal{E}}{L(\Gamma)},\gf). \]
\item \label{item_BPHZlocal}
The counterterms $C_{\Gamma}$ are local for every connected stable graph $\Gamma$, in fact
\[ C_{\Gamma}\in\singfun\otimes\Hom_{\mathrm{loc}}(\ltens{\mathcal{E}}{L(\Gamma)},\gf). \]
More precisely, we can write $C_{\Gamma}$ as a finite sum,
\[ C_{\Gamma} = \sum_{i=1}^N g_i\Psi_i, \]
where the $\Psi_i$ are local functionals and the $g_i$ are purely singular functions of $\varepsilon$ \emph{with finite order poles at zero}.
\end{enumerate}
\end{lemma}

\begin{rem} \label{rem_ctind}
Following our proof that the counterterms are independent of $L$, we will denote these counterterms simply by $C_{\Gamma}(\varepsilon)$. In the next section, we will use this independence to define an effective field theory.
\end{rem}

\clearpage

\begin{proof}
The proof is by induction on the number of edges of $\Gamma$, where we assume \eqref{item_BPHZnonsing}, \eqref{item_BPHZlind} and \eqref{item_BPHZlocal} hold for all graphs with fewer edges than $\Gamma$. If $\Gamma$ is a connected graph then \eqref{item_BPHZnonsing} follows from the defining equation \eqref{eqn_BPHZdef}. If $\Gamma$ is not connected then we may write $\Gamma=\Gamma_1\sqcup\Gamma_2$ for some stable graphs $\Gamma_1$, $\Gamma_2$. Then a simple and standard calculation, see for instance \cite[\S 5.3.3]{collins}, shows that
\[ \overline{R}_{\Gamma}(\varepsilon, L)+C_{\Gamma}(\varepsilon, L) = \left(\overline{R}_{\Gamma_1}(\varepsilon, L)+C_{\Gamma_1}(\varepsilon, L)\right) \otimes \left(\overline{R}_{\Gamma_2}(\varepsilon, L)+C_{\Gamma_2}(\varepsilon, L)\right). \]
Here we have used that $C_*(\varepsilon, L)$ is multiplicative and that every component of $\Gamma$ has at least one edge. By the inductive hypothesis, this expression converges as $\varepsilon\to 0$. This establishes \eqref{item_BPHZnonsing}.

To prove \eqref{item_BPHZlind}, it suffices to consider the case when $\Gamma$ is connected with a nonempty set of edges, for if $\Gamma$ has no edges then $C_{\Gamma}=0$. We fix an $L_0\in\mathbb{R}$ and calculate,
\begin{multline*}
\overline{R}_{\Gamma}(\varepsilon,L) = \sum_{\gamma\varsubsetneq\Gamma} F_{(\Gamma,\gamma;C_*(\varepsilon,L))}(I,P(\varepsilon,L)) = \sum_{\gamma\varsubsetneq\Gamma}F_{(\Gamma,\gamma;C_*(\varepsilon))}(I,P(\varepsilon,L_0)+P(L_0,L)) \\
= \sum_{\gamma\varsubsetneq\Gamma}\sum_{\gamma'\subset\Gamma/\gamma} \left[C_{\gamma}(\varepsilon)\bigotimes_{v\in V(\Gamma)-V(\gamma)} \left[\widetilde{I}_{g_v,|v|}\right] \right]\left(\bigotimes_{e\in E(\gamma')}\left[P(\varepsilon,L_0)\right]\bigotimes_{e\in E(\Gamma/\gamma)-E(\gamma')} \left[P(L_0, L)\right] \right),
\end{multline*}
where we evaluate the tensor between the left set of parentheses on the tensor between the right parentheses. Here we have already adopted the convention laid out in Remark \ref{rem_ctind} and made use of the inductive hypothesis. Note that \eqref{item_BPHZlocal} and the inductive hypothesis ensure, by Theorem \ref{thm_feynasy}, that $\overline{R}_{\Gamma}$ has the requisite asymptotic expansion.

We may rearrange the above sum by replacing subgraphs of $\Gamma/\gamma$ with subgraphs of $\Gamma$ containing $\gamma$.
\begin{multline*}
\overline{R}_{\Gamma}(\varepsilon,L) = \\
\sum_{\gamma'\subset\Gamma}\sum_{\begin{subarray}{c} \gamma\subset\gamma': \\ \gamma\neq\Gamma \end{subarray}}\left[C_{\gamma}(\varepsilon)\bigotimes_{v\in V(\Gamma)-V(\gamma)} \left[\widetilde{I}_{g_v,|v|}\right] \right] \left(\bigotimes_{e\in E(\gamma')-E(\gamma)}\left[P(\varepsilon,L_0)\right]\bigotimes_{e\in E(\Gamma)-E(\gamma')} \left[P(L_0, L)\right] \right) \\
\shoveleft{= \sum_{\gamma\varsubsetneq\Gamma} \left[C_{\gamma}(\varepsilon)\bigotimes_{v\in V(\Gamma)-V(\gamma)} \left[\widetilde{I}_{g_v,|v|}\right] \right]\left(\bigotimes_{e\in E(\Gamma)-E(\gamma)}\left[P(\varepsilon,L_0)\right]\right) +} \\
\sum_{\gamma'\varsubsetneq\Gamma}\sum_{\gamma\subset\gamma'} \left[C_{\gamma}(\varepsilon)\bigotimes_{v\in V(\Gamma)-V(\gamma)} \left[\widetilde{I}_{g_v,|v|}\right] \right] \left(\bigotimes_{e\in E(\gamma')-E(\gamma)}\left[P(\varepsilon,L_0)\right]\bigotimes_{e\in E(\Gamma)-E(\gamma')} \left[P(L_0, L)\right] \right) \\
\shoveleft{= \overline{R}_{\Gamma}(\varepsilon,L_0) +} \\
\sum_{\gamma'\varsubsetneq\Gamma}\sum_{\gamma\subset\gamma'} \left[C_{\gamma}(\varepsilon)\bigotimes_{v\in V(\Gamma)-V(\gamma)} \left[\widetilde{I}_{g_v,|v|}\right] \right] \left(\bigotimes_{e\in E(\gamma')-E(\gamma)}\left[P(\varepsilon,L_0)\right]\bigotimes_{e\in E(\Gamma)-E(\gamma')} \left[P(L_0, L)\right] \right).
\end{multline*}
Now we calculate the latter part of this expression as follows.
\begin{multline*}
\sum_{\gamma'\varsubsetneq\Gamma}\sum_{\gamma\subset\gamma'} \left[C_{\gamma}(\varepsilon)\bigotimes_{v\in V(\gamma')-V(\gamma)} \left[\widetilde{I}_{g_v,|v|}\right]\bigotimes_{v\in V(\Gamma)-V(\gamma')} \left[\widetilde{I}_{g_v,|v|}\right] \right] \\
\shoveright{\left(\bigotimes_{e\in E(\gamma')-E(\gamma)}\left[P(\varepsilon,L_0)\right]\bigotimes_{e\in E(\Gamma)-E(\gamma')} \left[P(L_0, L)\right] \right)} \\
\shoveleft{= F_{\Gamma}(I,P(L_0,L)) + \sum_{\emptyset\varsubsetneq\gamma'\varsubsetneq\Gamma} \left[C_{\gamma'}(\varepsilon)\bigotimes_{v\in V(\Gamma)-V(\gamma')} \left[\widetilde{I}_{g_v,|v|}\right] \right]\left(\bigotimes_{e\in E(\Gamma)-E(\gamma')} \left[P(L_0, L)\right] \right) +} \\
\sum_{\emptyset\varsubsetneq\gamma'\varsubsetneq\Gamma}\sum_{\gamma\varsubsetneq\gamma'} \left[F_{(\gamma',\gamma;C_*(\varepsilon))}(I,P(\varepsilon,L_0))\bigotimes_{v\in V(\Gamma)-V(\gamma')} \left[\widetilde{I}_{g_v,|v|}\right] \right]\left(\bigotimes_{e\in E(\Gamma)-E(\gamma')} \left[P(L_0, L)\right] \right) \\
\shoveleft{= F_{\Gamma}(I,P(L_0,L)) +} \\
\sum_{\emptyset\varsubsetneq\gamma'\varsubsetneq\Gamma} \left[\left(\overline{R}_{\gamma'}(\varepsilon,L_0)+C_{\gamma'}(\varepsilon)\right)\bigotimes_{v\in V(\Gamma)-V(\gamma')} \left[\widetilde{I}_{g_v,|v|}\right] \right]\left(\bigotimes_{e\in E(\Gamma)-E(\gamma')} \left[P(L_0, L)\right] \right).
\end{multline*}
Combining these two calculations we arrive at,
\begin{equation} \label{eqn_lind}
\overline{R}_{\Gamma}(\varepsilon,L) = \overline{R}_{\Gamma}(\varepsilon,L_0) + F_{\Gamma}(I,P(L_0,L)) + \sum_{\emptyset\varsubsetneq\gamma'\varsubsetneq\Gamma} F_{(\Gamma,\gamma';\overline{R}_*(\varepsilon,L_0)+C_*(\varepsilon))}(I,P(L_0,L)).
\end{equation}
By \eqref{item_BPHZnonsing}, the expression $\overline{R}_{\gamma'}(\varepsilon,L_0)+C_{\gamma'}(\varepsilon)$ converges as $\varepsilon\to 0$, so all the terms of \eqref{eqn_lind} are nonsingular except possibly the first. Hence,
\[ C_{\Gamma}(-,-) = -\Sing\left[\overline{R}_{\Gamma}(-,-)\right] = -\Sing\left[\overline{R}_{\Gamma}(-,L_0)\right]. \]
Since $\overline{R}_{\Gamma}(-,L_0)$ does not depend on the parameter $L$ as we have fixed $L=L_0$, it follows from the naturality of the operator $\Sing$ that its singular part is also independent of $L$. This shows that the counterterms $C_{\Gamma}(\varepsilon,L)$ do not depend on $L$, which finishes the proof of \eqref{item_BPHZlind}.

Now \eqref{item_BPHZlocal} will follow as a simple consequence of \eqref{item_BPHZlind}. By part \eqref{item_feynasyexp} of Theorem \ref{thm_feynasy} we may write the counterterm $C_{\Gamma}$ as a finite sum,
\[ C_{\Gamma} = \sum_{i=1}^N g_i\Psi_i; \qquad g_i\in\singfun, \Psi_i\in\Hom(\ltens{\mathcal{E}}{L(\Gamma)},\smth{0,\infty}); \]
where the $g_i$ have finite order poles at zero. Since the counterterm $C_{\Gamma}$ is independent of $L$, we may assume that each $\Psi_i$ does not depend on $L$ either. From this and part \eqref{item_feynasyloc} of Theorem \ref{thm_feynasy} we may write,
\[ \Psi_i = \lim_{L\to 0}\left[\sum_{j=1}^K f_{ij}(L)\psi_{ij}\right], \]
where each $\psi_{ij}\in\Hom_{\mathrm{loc}}(\ltens{\mathcal{E}}{L(\Gamma)},\gf)$ and the $f_{ij}$ are smooth functions of $L$. Since this limit occurs in a finite-dimensional and hence closed subspace of a Hausdorff topological vector space, we may conclude that each $\Psi_i$ is itself a local functional. This completes the proof of \eqref{item_BPHZlocal} and this lemma.
\end{proof}

We recall that the counterterms $C_{\Gamma}$ are used to define the renormalized amplitude of a stable graph, whose definition is as follows.

\begin{defi}
Let $I\in\psloc{\mathcal{E}}[[\hbar]]$ be a local interaction. The renormalized Feynman amplitude of a connected stable graph $\Gamma$ at length scale $L$ is defined by
\[ R_{\Gamma}(L):=\lim_{\varepsilon\to 0}\left[\overline{R}_{\Gamma}(\varepsilon,L)+C_{\Gamma}(\varepsilon)\right]\in\ltens{\mathcal{E}^{\dag}}{L(\Gamma)}, \qquad 0<L\leq\infty. \]
\end{defi}
We will frequently abuse notation and denote the corresponding monomial in $\psloc{\mathcal{E}}$ by the same symbol $R_{\Gamma}(L)$, just as we did with the usual Feynman amplitude in Definition \ref{def_feynamp}. We will adopt the same convention for $C_{\Gamma}(\varepsilon)$ and $\overline{R}_{\Gamma}(\varepsilon,L)$ as well.

We state one final simple lemma before we move on to the proof of our main theorem.

\begin{lemma} \label{lem_cteqvar}
The functionals $C_*(\varepsilon)$ and $\overline{R}_*(\varepsilon,L)$ on stable graphs are equivariant.
\end{lemma}

\begin{proof}
The result follows from a simple induction on the number of edges in a graph.
\end{proof}

\section{Main theorem on the construction of an effective field theory} \label{sec_main}

In this section we define a family of effective interactions through the graphical combinatorics of the BPHZ algorithm. This family of interactions will form an effective field theory satisfying the axioms of Definition \ref{def_effthy}. We prove that this effective field theory defined through the BPHZ algorithm is the same as the effective field theory \eqref{eqn_effthy} that was defined by Costello and described in Section \ref{sec_effthyconst}. We recall that the latter required no graphical combinatorics for its construction, relying only on the simple construction \eqref{eqn_counterterms} of counterterms described in Theorem \ref{thm_counterterms}.

We begin by defining this family of effective interactions using the renormalized amplitude of a stable graph.

\begin{defi}
Let $I\in\psloc{\mathcal{E}}[[\hbar]]$ be a local interaction and define a family of effective interactions by
\begin{equation} \label{eqn_BPHZeffthy}
I[L] := \sum_{i=0}^{\infty} \hbar^iI_{i0} + \sum_{\Gamma}\frac{\hbar^{g(\Gamma)}}{|\Aut(\Gamma)|} R_{\Gamma}(L), \qquad 0<L\leq\infty;
\end{equation}
where we sum over all connected stable graphs $\Gamma$.
\end{defi}

With some effort, we may show directly that the effective interactions defined by Equation \eqref{eqn_BPHZeffthy} form an effective field theory. However, our main theorem which we will now state and prove will show that this family of effective interactions is the same as the one defined by Equation \eqref{eqn_effthy}, and this family of interactions forms an effective field theory by Theorem \ref{thm_effthy}. Hence, one immediate consequence of our main theorem will be that the effective interactions defined by \eqref{eqn_BPHZeffthy} above form an effective field theory in the sense of Definition \ref{def_effthy}.

\begin{theorem} \label{thm_main}
Let $I\in\psloc{\mathcal{E}}[[\hbar]]$ be a local interaction.
\begin{enumerate}
\item \label{item_maincterms}
Consider the counterterms $I_{ij}^{CT}$ defined in Theorem \ref{thm_counterterms} by Equation \eqref{eqn_counterterms} and the counterterms $C_{\Gamma}$ defined through the BPHZ algorithm by Equation \eqref{eqn_BPHZdef} of Definition \ref{def_BPHZ}. Then,
\begin{equation} \label{eqn_ctermeqv}
I_{ij}^{CT} = -\sum_{\Gamma: \begin{subarray}{c} g(\Gamma)=i \\ |L(\Gamma)|=j \end{subarray}} \frac{1}{|\Aut(\Gamma)|} C_{\Gamma}, \qquad i,j\geq 0;
\end{equation}
where we sum over all connected stable graphs of genus $i$ with $j$ legs.
\item \label{item_maineffthy}
Consider the family of effective interactions \eqref{eqn_effthy} defined by Costello through his construction of the counterterms $I_{ij}^{CT}$ and the family of effective interactions \eqref{eqn_BPHZeffthy} defined through the BPHZ algorithm. Then these two families of effective interactions are identical, hence these two constructions define the same effective field theory.
\end{enumerate}
\end{theorem}

\begin{rem}
We mention that the proof of Equation \eqref{eqn_ctermeqv} involves the same graphical combinatorics that appear in the well-known renormalization group identity,
\[ W(W(I,P_1),P_2) = W(I,P_1+P_2). \]
This particular identity is usually proved by making use of formula \eqref{eqn_feynographs}, see for instance Lemma 3.4.1 of \cite[\S 2.3]{coseffthy}.
\end{rem}

\begin{proof}
We begin with the proof of \eqref{item_maincterms}, which will be by induction along the well-ordering of the index set that is defined by \eqref{eqn_indexorder}. In fact, we will use induction to prove the identity,
\begin{equation} \label{eqn_mainid}
W_{ij}\left(I-\sum_{(p,q)<(i,j)}\hbar^p I_{pq}^{CT}(\varepsilon),P(\varepsilon,L)\right) = I_{ij} + \sum_{\Gamma: \begin{subarray}{c} g(\Gamma)=i \\ \ |L(\Gamma)|=j \\ |E(\Gamma)|>0 \end{subarray}} \frac{1}{|\Aut(\Gamma)|}\overline{R}_{\Gamma}(\varepsilon,L), \qquad i,j\geq 0;
\end{equation}
where we sum over all connected stable graphs of genus $i$ with $j$ legs that have at least one edge. Equation \eqref{eqn_ctermeqv} then follows as an immediate consequence of this identity by applying the operator $\Sing$ to both sides, which yields Costello's counterterms on the left and the counterterms coming from BPHZ on the right.

The calculation establishing \eqref{eqn_mainid} will be quite involved, so we start by introducing some notation. Write down a list,
\[ \Gamma_{\alpha}, \quad \alpha\in\mathcal{I}_{pq}; \]
of representatives for isomorphism classes of connected stable graphs of genus $p$ with $q$ legs and at least one edge. Similarly, let
\[ \left(\Gamma_{\alpha(\lambda)},\tau_{\lambda},\gamma_{\lambda}\right), \quad \lambda\in\mathcal{J}_{pq} \quad\qquad\text{and}\quad\qquad \left(\Gamma_{\alpha(\kappa)},\gamma_{\kappa}\right), \quad \kappa\in\mathcal{K}_{pq} \]
be lists of representatives for the isomorphism classes of those objects in $\ctrips$ and $\cpairs$ respectively for which $\Gamma$ has genus $p$ with $q$ legs. Note that Theorem \ref{thm_eqvcat} provides a bijection between the sets $\mathcal{J}_{pq}$ and $\mathcal{K}_{pq}$.

We begin by calculating the left-hand side of \eqref{eqn_mainid}. By Equation \eqref{eqn_feycomp},
\begin{multline*}
W_{ij}\left(I-\sum_{(p,q)<(i,j)}\hbar^p I_{pq}^{CT}(\varepsilon),P(\varepsilon,L)\right) = \\
I_{ij} + \sum_{\alpha\in\mathcal{I}_{ij}} \frac{1}{|\Aut(\Gamma_{\alpha})|}F_{\Gamma_{\alpha}}\left(I-\sum_{(p,q)<(i,j)}\hbar^p I_{pq}^{CT}(\varepsilon),P(\varepsilon,L)\right).
\end{multline*}
We focus on calculating the Feynman amplitude in this expression. Note that if $\Gamma$ is a stable graph of genus $i$ with $j$ legs and at least one edge, then an elementary argument shows that for every vertex $v$ of $\Gamma$, $(g_v,|v|)<(i,j)$. This means that,
\begin{multline*}
F_{\Gamma_{\alpha}}\left(I-\sum_{(p,q)<(i,j)}\hbar^p I_{pq}^{CT}(\varepsilon),P(\varepsilon,L)\right) = \\
\left[\bigotimes_{v\in V(\Gamma_{\alpha})}\left[\widetilde{I}_{g_v,|v|}-\widetilde{I}^{CT}_{g_v,|v|}(\varepsilon)\right]\right]\left(\bigotimes_{e\in E(\Gamma_{\alpha})}\left[P(\varepsilon,L)\right]\right) \\
= \sum_{\mathscr{V}\subset V(\Gamma_{\alpha})} \left[\bigotimes_{v\in V(\Gamma_{\alpha})-\mathscr{V}}\left[\widetilde{I}_{g_v,|v|}\right]\bigotimes_{v\in\mathscr{V}}\left[-\widetilde{I}^{CT}_{g_v,|v|}(\varepsilon)\right]\right] \left(\bigotimes_{e\in E(\Gamma_{\alpha})}\left[P(\varepsilon,L)\right]\right) \\
= \sum_{\mathscr{V}\subset V(\Gamma_{\alpha})} \left[\bigotimes_{v\in V(\Gamma_{\alpha})-\mathscr{V}}\left[\widetilde{I}_{g_v,|v|}\right]\bigotimes_{v\in\mathscr{V}}\left[\sum_{\beta\in\mathcal{I}_{g_v,|v|}}\frac{1}{|\Aut(\Gamma_{\beta})|}\widetilde{C}_{\Gamma_{\beta}}(\varepsilon)\right]\right] \left(\bigotimes_{e\in E(\Gamma_{\alpha})}\left[P(\varepsilon,L)\right]\right),
\end{multline*}
where we have taken the sum over all subsets $\mathscr{V}$ of $V(\Gamma_{\alpha})$ and used the inductive hypothesis.

If we examine the term $\widetilde{C}_{\Gamma_{\beta}}(\varepsilon)\in\ltens{\mathcal{E}^{\dag}}{v}$ in the above expression we find,
\[ \widetilde{C}_{\Gamma_{\beta}}(\varepsilon) = \sum_{\sigma\in\Bij\left(L(\Gamma_{\beta}),v\right)} \sigma^{\#} \left(C_{\Gamma_{\beta}}(\varepsilon)\right), \]
where we sum over all bijections between the legs of the graph and the vertex.

Before computing the Feynman amplitude further, we first clarify some new notation. We denote by $\prod_{v\in\mathscr{V}}\mathcal{I}_{g_v,|v|}$ the set of all those functions $\beta$ on $\mathscr{V}$ that assign an index $\beta_v\in\mathcal{I}_{g_v,|v|}$ to a vertex $v\in\mathscr{V}$. Likewise, we denote by $\prod_{v\in\mathscr{V}}\Bij(L(\Gamma_{\beta_v}),v)$ the set of all those functions $\sigma$ that assign to every vertex $v\in\mathscr{V}$, a bijection $\sigma_v$ from $L(\Gamma_{\beta_v})$ to $v$.

Now we continue our computation by distributing the sums,
\begin{multline*}
F_{\Gamma_{\alpha}}\left(I-\sum_{(p,q)<(i,j)}\hbar^p I_{pq}^{CT}(\varepsilon),P(\varepsilon,L)\right) = \\
\sum_{\mathscr{V}\subset V(\Gamma_{\alpha})}\sum_{\beta\in\underset{v\in\mathscr{V}}{\prod}\mathcal{I}_{g_v,|v|}}\frac{1}{\underset{v\in\mathscr{V}}{\prod}|\Aut(\Gamma_{\beta_v})|} \left[\bigotimes_{v\in V(\Gamma_{\alpha})-\mathscr{V}}\left[\widetilde{I}_{g_v,|v|}\right]\bigotimes_{v\in\mathscr{V}}\left[\widetilde{C}_{\Gamma_{\beta_v}}(\varepsilon)\right]\right] \left(\bigotimes_{e\in E(\Gamma_{\alpha})}\left[P(\varepsilon,L)\right]\right) \\
\shoveleft{= \sum_{\mathscr{V}\subset V(\Gamma_{\alpha})} \ \sum_{\beta\in\underset{v\in\mathscr{V}}{\prod}\mathcal{I}_{g_v,|v|}} \ \sum_{\sigma\in\underset{v\in\mathscr{V}}{\prod}\Bij\left(L(\Gamma_{\beta_v}),v\right)}} \\
\frac{1}{\underset{v\in\mathscr{V}}{\prod}|\Aut(\Gamma_{\beta_v})|} \left[\bigotimes_{v\in V(\Gamma_{\alpha})-\mathscr{V}}\left[\widetilde{I}_{g_v,|v|}\right]\bigotimes_{v\in\mathscr{V}}\left[\sigma_v^{\#}\left(C_{\Gamma_{\beta_v}}(\varepsilon)\right)\right]\right] \left(\bigotimes_{e\in E(\Gamma_{\alpha})}\left[P(\varepsilon,L)\right]\right).
\end{multline*}

Given a connected stable graph $\Gamma$ with at least one edge, a subset $\mathscr{V}\subset V(\Gamma)$ and $\beta$ and $\sigma$ as above we may define a map
\[ \tau_{(\Gamma,\mathscr{V},\beta,\sigma)} := \left(\bigsqcup_{v\in\mathscr{V}} \sigma_v: L\bigg(\bigsqcup_{v\in\mathscr{V}}\Gamma_{\beta_v}\bigg)\to H(\Gamma)\right). \]
This map satisfies the requirements of Definition \ref{def_graphinsert} so that $(\Gamma,\tau_{(\Gamma,\mathscr{V},\beta,\sigma)},\sqcup_{v\in\mathscr{V}}\Gamma_{\beta_v})$ defines an object of $\ctrips$.

The next step is to split up our sum using the index set $\mathcal{J}_{ij}$. For this we introduce the notation $\prod_{c}\big|\Aut\left(\gamma^{(c)}\right)\big|$ for the product of the cardinalities of the automorphism groups of the connected components of a stable graph $\gamma$. Combining our calculations so far yields,
\begin{multline*}
W_{ij}\left(I-\sum_{(p,q)<(i,j)}\hbar^p I_{pq}^{CT}(\varepsilon),P(\varepsilon,L)\right) = I_{ij} + \sum_{\alpha\in\mathcal{I}_{ij}} \ \sum_{\mathscr{V}\subset V(\Gamma_{\alpha})} \ \sum_{\beta\in\underset{v\in\mathscr{V}}{\prod}\mathcal{I}_{g_v,|v|}} \ \sum_{\sigma\in\underset{v\in\mathscr{V}}{\prod}\Bij\left(L(\Gamma_{\beta_v}),v\right)} \\
\frac{1}{|\Aut(\Gamma_{\alpha})|\underset{v\in\mathscr{V}}{\prod}|\Aut(\Gamma_{\beta_v})|} \left[\bigotimes_{v\in V(\Gamma_{\alpha})-\mathscr{V}}\left[\widetilde{I}_{g_v,|v|}\right]\bigotimes\tau_{(\Gamma_{\alpha},\mathscr{V},\beta,\sigma)}^{\#}\left(C_{\underset{v\in\mathscr{V}}{\sqcup}\Gamma_{\beta_v}}(\varepsilon)\right)\right] \left(\bigotimes_{e\in E(\Gamma_{\alpha})}\left[P(\varepsilon,L)\right]\right).
\end{multline*}
Now we notice that the terms in this sum depend only on the isomorphism class of
\[ \left(\Gamma_{\alpha},\tau_{(\Gamma_{\alpha},\mathscr{V},\beta,\sigma)},\underset{v\in\mathscr{V}}{\sqcup}\Gamma_{\beta_v}\right) \]
in $\ctrips$, which follows from Lemma \ref{lem_cteqvar}. Hence,
\begin{multline} \label{eqn_mainlhscount}
W_{ij}\left(I-\sum_{(p,q)<(i,j)}\hbar^p I_{pq}^{CT}(\varepsilon),P(\varepsilon,L)\right) = I_{ij} + \\
\sum_{\lambda\in\mathcal{J}_{ij}} \frac{|E_{\lambda}|}{\Big|\Aut\big(\Gamma_{\alpha(\lambda)}\big)\Big|\prod_c \Big|\Aut\big(\gamma_{\lambda}^{(c)}\big)\Big|} \left[\bigotimes_{v\in V(\Gamma_{\alpha(\lambda)})-\mathscr{V_{\gamma_{\lambda}}}}\left[\widetilde{I}_{g_v,|v|}\right]\bigotimes\tau_{\lambda}^{\#}\left(C_{\gamma_{\lambda}}(\varepsilon)\right)\right] \left(\bigotimes_{e\in E(\Gamma_{\alpha(\lambda)})}\left[P(\varepsilon,L)\right]\right);
\end{multline}
where $E_{\lambda}$ is the set consisting of all 3-tuples $(\mathscr{V},\beta,\sigma)$ where $\mathscr{V}$ is a subset of $V(\Gamma_{\alpha(\lambda)})$ and $\beta$ and $\sigma$ are as above and for which,
\[ \Big(\Gamma_{\alpha(\lambda)},\tau_{(\Gamma_{\alpha(\lambda)},\mathscr{V},\beta,\sigma)},\underset{v\in\mathscr{V}}{\sqcup}\Gamma_{\beta_v}\Big)\cong\Big(\Gamma_{\alpha(\lambda)},\tau_{\lambda},\gamma_{\lambda}\Big). \]

It remains to count the elements of $E_{\lambda}$. We may assume that we can write $\gamma_{\lambda}$ as a union,
\[ \gamma_{\lambda} = \Gamma_{\alpha_1}\sqcup\Gamma_{\alpha_2}\sqcup\ldots\sqcup\Gamma_{\alpha_k} \]
for some choice of indices $\alpha_1,\ldots,\alpha_k$. Let $G_{\lambda}$ denote the set of all those injective functions,
\[ \tau:L(\gamma_{\lambda})\to H(\Gamma_{\alpha(\lambda)}) \]
satisfying conditions \eqref{item_legstovert} and \eqref{item_genusglue} of Definition \ref{def_graphinsert} and such that,
\[ \big(\Gamma_{\alpha(\lambda)},\tau,\gamma_{\lambda}\big)\cong\big(\Gamma_{\alpha(\lambda)},\tau_{\lambda},\gamma_{\lambda}\big). \]
If $\tau$ is such a function then set $v^{\tau}_i:=\tau(L(\Gamma_{\alpha_i}))\in V(\Gamma_{\alpha(\lambda)})$. We may define a projection map,
\[ \pi:G_{\lambda}\to E_{\lambda}, \qquad \tau\mapsto (\mathscr{V}^{\tau},\beta^{\tau},\sigma^{\tau}); \]
in the following obvious fashion:
\begin{displaymath}
\begin{split}
\mathscr{V}^{\tau} &:= \{v_i^{\tau}:i=1,\ldots,k\}, \\
\beta^{\tau}_{v_i^{\tau}} & := \alpha_i, \\
\sigma^{\tau}_{v_i^{\tau}} & :=\left(\tau_{|L(\Gamma_{\alpha_i})}:L(\Gamma_{\alpha_i})\to v_i^{\tau}\right).
\end{split}
\end{displaymath}

Given a point $(\mathscr{V},\beta,\sigma)$ in $E_{\lambda}$, we may describe the (always nonempty) fiber $\pi^{-1}(\mathscr{V},\beta,\sigma)$ as follows. Consider the subgroup of $S_k$ defined by
\[ \Lambda_k:=\left\{\varsigma\in S_k: \alpha_{\varsigma(i)}=\alpha_i\text{ for all }i=1,\ldots, k \right\}. \]
This group acts on the fiber $\pi^{-1}(\mathscr{V},\beta,\sigma)$ in the obvious way by simply permuting the components of $\gamma_{\lambda}$. Checking carefully, we see that this action is free and transitive. Hence $G_{\lambda}$ is a (discrete) principal $\Lambda_k$-bundle over $E_{\lambda}$. It follows that,
\[ |G_{\lambda}|=|E_{\lambda}||\Lambda_k|. \]

Consider the group $\Aut(\Gamma_{\alpha(\lambda)})\times\Aut(\gamma_{\lambda})$, which acts transitively on $G_{\lambda}$ in the obvious way. The isotropy subgroup of $\tau_{\lambda}$ is $\Aut(\Gamma_{\alpha(\lambda)},\tau_{\lambda},\gamma_{\lambda})$. Hence,
\[ |G_{\lambda}| = \frac{|\Aut(\Gamma_{\alpha(\lambda)})||\Aut(\gamma_{\lambda})|}{|\Aut(\Gamma_{\alpha(\lambda)},\tau_{\lambda},\gamma_{\lambda})|}. \]
Since,
\[ \Aut(\gamma_{\lambda})\cong\left(\Aut(\Gamma_{\alpha_1})\times\cdots\times\Aut(\Gamma_{\alpha_k})\right)\rtimes\Lambda_k, \]
we compute,
\begin{equation} \label{eqn_count}
|E_{\lambda}| = \frac{|G_{\lambda}|}{|\Lambda_k|} = \frac{|\Aut(\Gamma_{\alpha(\lambda)})||\Aut(\Gamma_{\alpha_1})|\cdots|\Aut(\Gamma_{\alpha_k})|}{|\Aut(\Gamma_{\alpha(\lambda)},\tau_{\lambda},\gamma_{\lambda})|}.
\end{equation}
Substituting \eqref{eqn_count} into \eqref{eqn_mainlhscount} we finally arrive at the identity,
\begin{multline} \label{eqn_mainidlhs}
W_{ij}\left(I-\sum_{(p,q)<(i,j)}\hbar^p I_{pq}^{CT}(\varepsilon),P(\varepsilon,L)\right) = I_{ij} + \\
\sum_{\lambda\in\mathcal{J}_{ij}} \frac{1}{\big|\Aut(\Gamma_{\alpha(\lambda)},\tau_{\lambda},\gamma_{\lambda})\big|} \left[\bigotimes_{v\in V(\Gamma_{\alpha(\lambda)})-\mathscr{V_{\gamma_{\lambda}}}}\left[\widetilde{I}_{g_v,|v|}\right]\bigotimes\tau_{\lambda}^{\#}\left(C_{\gamma_{\lambda}}(\varepsilon)\right)\right] \left(\bigotimes_{e\in E(\Gamma_{\alpha(\lambda)})}\left[P(\varepsilon,L)\right]\right).
\end{multline}

Having finished our calculation of the left-hand side of \eqref{eqn_mainid}, we now proceed with our calculation of the right-hand side, which will be somewhat shorter.
\begin{displaymath}
\begin{split}
\sum_{\Gamma: \begin{subarray}{c} g(\Gamma)=i \\ \ |L(\Gamma)|=j \\ |E(\Gamma)|>0 \end{subarray}} \frac{1}{|\Aut(\Gamma)|}\overline{R}_{\Gamma}(\varepsilon,L) &= \sum_{\alpha\in\mathcal{I}_{ij}} \frac{1}{|\Aut(\Gamma_{\alpha})|}\overline{R}_{\Gamma_{\alpha}}(\varepsilon,L) \\
&= \sum_{\alpha\in\mathcal{I}_{ij}}\sum_{\gamma\varsubsetneq\Gamma_{\alpha}} \frac{1}{|\Aut(\Gamma_{\alpha})|} F_{(\Gamma_{\alpha},\gamma,C_*(\varepsilon))}(I,P(\varepsilon,L)).
\end{split}
\end{displaymath}
By Lemma \ref{lem_cteqvar} we see that the terms in this sum depend only on the isomorphism class of $(\Gamma_{\alpha},\gamma)$ in $\cpairs$. Splitting up this sum using the index set $\mathcal{K}_{ij}$ we arrive at,
\begin{equation} \label{eqn_mainrhscount}
\sum_{\alpha\in\mathcal{I}_{ij}} \frac{1}{|\Aut(\Gamma_{\alpha})|}\overline{R}_{\Gamma_{\alpha}}(\varepsilon,L) = \sum_{\kappa\in\mathcal{K}_{ij}} \frac{|B_{\kappa}|}{\big|\Aut(\Gamma_{\alpha(\kappa)})\big|} F_{(\Gamma_{\alpha(\kappa)},\gamma_{\kappa},C_*(\varepsilon))}(I,P(\varepsilon,L)),
\end{equation}
where $B_{\kappa}$ is the set consisting of all those proper subgraphs $\gamma$ of $\Gamma_{\alpha(\kappa)}$ such that
\[ (\Gamma_{\alpha(\kappa)},\gamma)\cong(\Gamma_{\alpha(\kappa)},\gamma_{\kappa}). \]

Since $\Aut(\Gamma_{\alpha(\kappa)})$ acts transitively on $B_{\kappa}$ and since the isotropy subgroup of $\gamma_{\kappa}$ is $\Aut(\Gamma_{\alpha(\kappa)},\gamma_{\kappa})$ we conclude that,
\[ |B_{\kappa}| = \frac{|\Aut(\Gamma_{\alpha(\kappa)})|}{|\Aut(\Gamma_{\alpha(\kappa)},\gamma_{\kappa})|}. \]
Substituting this expression into \eqref{eqn_mainrhscount} yields,
\begin{equation} \label{eqn_mainidrhs}
\sum_{\alpha\in\mathcal{I}_{ij}} \frac{1}{|\Aut(\Gamma_{\alpha})|}\overline{R}_{\Gamma_{\alpha}}(\varepsilon,L) = \sum_{\kappa\in\mathcal{K}_{ij}} \frac{1}{|\Aut(\Gamma_{\alpha(\kappa)},\gamma_{\kappa})|} F_{(\Gamma_{\alpha(\kappa)},\gamma_{\kappa},C_*(\varepsilon))}(I,P(\varepsilon,L)).
\end{equation}

Comparing our expression \eqref{eqn_mainidlhs} for the left-hand side of Equation \eqref{eqn_mainid} with the above expression \eqref{eqn_mainidrhs} for the second term on the right-hand side of Equation \eqref{eqn_mainid}, we see that Equation \eqref{eqn_mainid} is equivalent to the identity,
\begin{multline*}
\sum_{\lambda\in\mathcal{J}_{ij}} \frac{1}{\big|\Aut(\Gamma_{\alpha(\lambda)},\tau_{\lambda},\gamma_{\lambda})\big|} \left[\bigotimes_{v\in V(\Gamma_{\alpha(\lambda)})-\mathscr{V_{\gamma_{\lambda}}}}\left[\widetilde{I}_{g_v,|v|}\right]\bigotimes\tau_{\lambda}^{\#}\left(C_{\gamma_{\lambda}}(\varepsilon)\right)\right] \left(\bigotimes_{e\in E(\Gamma_{\alpha(\lambda)})}\left[P(\varepsilon,L)\right]\right) \\
= \sum_{\kappa\in\mathcal{K}_{ij}} \frac{1}{|\Aut(\Gamma_{\alpha(\kappa)},\gamma_{\kappa})|} F_{(\Gamma_{\alpha(\kappa)},\gamma_{\kappa},C_*(\varepsilon))}(I,P(\varepsilon,L)).
\end{multline*}
Now the above identity follows from Theorem \ref{thm_eqvcat}, which provides a bijection between the terms in these two sums. This concludes the proof of Equation \eqref{eqn_mainid} and the first part of this theorem.

We now prove \eqref{item_maineffthy}, which follows quite easily now that we have established Equation \eqref{eqn_mainid}. To show that the effective field theory defined by \eqref{eqn_effthy} is the same as the effective field theory defined by \eqref{eqn_BPHZeffthy} is equivalent to proving the identity,
\begin{equation} \label{eqn_effthyeqv}
\lim_{\varepsilon\to 0}\left[ W_{ij}\left(I-\sum_{p,q=0}^{\infty}\hbar^p I_{pq}^{CT}(\varepsilon),P(\varepsilon,L)\right) \right] = I_{ij} + \sum_{\alpha\in\mathcal{I}_{ij}}\frac{1}{|\Aut(\Gamma_{\alpha})|}R_{\Gamma_{\alpha}}(L).
\end{equation}

By Lemma \ref{lem_basicid}, Equation \eqref{eqn_mainid} and Equation \eqref{eqn_ctermeqv} we see that,
\begin{displaymath}
\begin{split}
W_{ij}\left(I-\sum_{p,q=0}^{\infty}\hbar^p I_{pq}^{CT}(\varepsilon),P(\varepsilon,L)\right) &= W_{ij}\left(I-\sum_{(p,q)<(i,j)}\hbar^p I_{pq}^{CT}(\varepsilon),P(\varepsilon,L)\right) - I_{ij}^{CT}(\varepsilon) \\
&= I_{ij} + \sum_{\alpha\in\mathcal{I}_{ij}} \left[\frac{1}{|\Aut(\Gamma_{\alpha})|}\overline{R}_{\Gamma_{\alpha}}(\varepsilon,L)\right] - I_{ij}^{CT}(\varepsilon) \\
&= I_{ij} + \sum_{\alpha\in\mathcal{I}_{ij}} \frac{1}{|\Aut(\Gamma_{\alpha})|}\left[\overline{R}_{\Gamma_{\alpha}}(\varepsilon,L) + C_{\Gamma_{\alpha}}(\varepsilon)\right].
\end{split}
\end{displaymath}
Taking the limit as $\varepsilon\to 0$ of both sides yields \eqref{eqn_effthyeqv}. This concludes our proof of \eqref{item_maineffthy} and this theorem.
\end{proof}


\begin{thebibliography}{12}
%
\bibitem{bogpar}
N. Bogoliubov, O. Parasiuk; \emph{\"Uber die Multiplikation der Kausalfunktionen in der Quantentheorie der Felder.} (German) Acta Math. 97 1957 227--266.
%
\bibitem{collins}
J. C. Collins, \emph{Renormalization. An introduction to renormalization, the renormalization group, and the operator-product expansion.} Cambridge Monographs on Mathematical Physics. Cambridge University Press, Cambridge, 1984.
%
\bibitem{ckhopfI}
A. Connes, D. Kreimer; \emph{Renormalization in quantum field theory and the Riemann-Hilbert problem. I. The Hopf algebra structure of graphs and the main theorem.} Comm. Math. Phys. 210 (2000), no. 1, 249--273.
%
\bibitem{ckhopfII}
A. Connes, D. Kreimer; \emph{Renormalization in quantum field theory and the Riemann-Hilbert problem. II. The $\beta$-function, diffeomorphisms and the renormalization group.} Comm. Math. Phys. 216 (2001), no. 1, 215--241.
%
\bibitem{ckinsert}
A. Connes, D. Kreimer; \emph{Insertion and elimination: the doubly infinite Lie algebra of Feynman graphs.} Ann. Henri Poincar\'e  3 (2002), no. 3, 411--433.
%
\bibitem{costbv}
K. Costello, \emph{Renormalisation and the Batalin-Vilkovisky formalism.} \texttt{arXiv:0706.1533}.
%
\bibitem{coseffthy}
K. Costello, \emph{Renormalization and effective field theory.} Mathematical Surveys and Monographs, 170. American Mathematical Society, Providence, RI, 2011.
%
\bibitem{cycop}
E. Getzler, M. Kapranov; \emph{Cyclic operads and cyclic homology.} Geometry, topology and physics, 167--201, Conf. Proc. Lecture Notes Geom. Topology, IV, Int. Press, Cambridge, MA, 1995.
%
\bibitem{modop}
E. Getzler, M. Kapranov; \emph{Modular operads.} Compositio Math. 110 (1998), no. 1, 65--126.
%
\bibitem{hepp}
K. Hepp, \emph{Proof of the Bogoliubov-Parasiuk theorem on renormalization.} Comm. Math. Phys. Volume 2, Number 1 (1966), 301--326.
%
\bibitem{opalgtophys}
M. Markl, S. Shnider, J. Stasheff; \emph{Operads in algebra, topology and physics.} Mathematical Surveys and Monographs, 96. American Mathematical Society, Providence, RI, 2002.
%
\bibitem{review}
H. Nishimura, \emph{Zentralblatt MATH review of `Renormalization and effective field theory'.} \texttt{Zbl 1221.81004.}
%
\bibitem{treves}
F. Tr\'eves, \emph{Topological vector spaces, distributions and kernels.} Academic Press, New York-London 1967.
%
\bibitem{wilson}
K. Wilson, \emph{Renormalization group and critical phenomena. I. Renormalization group and the Kadanoff scaling picture.} Phys. Rev. B 4, (1971) 3174--3183.
%
\bibitem{zimmer}
W. Zimmermann, \emph{Convergence of Bogoliubov's method of renormalization in momentum space.} Comm. Math. Phys. 15 1969 208--234.
%
\end{thebibliography}
\end{document}